\documentclass[12pt, reqno]{amsart}

\usepackage{a4wide}
\usepackage{amsmath, amsfonts, amssymb}
\usepackage{cite}
\usepackage{graphicx}
\usepackage{hyperref}
\usepackage{setspace}
\usepackage{caption}
\usepackage{subcaption}
\usepackage[utf8]{inputenc}

\newcommand{\R}{\mathbb{R}}
\newcommand{\Rt}{\mathbb{R}^3}
\newcommand{\Om}{\Omega}
\newcommand{\Omt}{\Omega\times [0,T]}
\newcommand{\G}{\Gamma}
\newcommand{\Gt}{\Gamma\times [0,T]}
\newcommand{\Gin}{\Gamma_{\mathrm{in}}}
\newcommand{\Gout}{\Gamma_{\mathrm{out}}}
\newcommand{\bu}{\mathbf{u}}
\newcommand{\bv}{\mathbf{v}}
\newcommand{\s}{\sigma}
\newcommand{\e}{\varepsilon}
\newcommand{\pd}{\partial}
\newcommand{\na}{\nabla}
\newcommand{\bn}{\mathbf{n}}
\newcommand{\bF}{\mathbf{F}}
\newcommand{\D}{\Delta}

\newcommand{\bc}{\mathbf{c}}
\newcommand{\be}{\mathbf{e}}
\newcommand{\bx}{\mathbf{x}}
\newcommand{\dif}{\mathrm{d}}

\newcommand{\bnt}{\mathbf{0}}
\newcommand{\hnt}{(0,0,0)^T}
\newcommand{\bve}{\boldsymbol{\varepsilon}}
\newcommand{\bsi}{\boldsymbol{\sigma}}
\newcommand{\la}{\lambda}
\newcommand{\al}{\alpha}
\newcommand{\g}{\gamma}

\newtheorem*{thm}{Theorem}

\DeclareMathOperator{\Def}{def}
\DeclareMathOperator{\Det}{det}

\numberwithin{equation}{section}

\begin{document}

\title[A 2-D model of curvilinear blood vessels with layered elastic walls]{A two dimensional model of curvilinear blood vessels with layered elastic walls}

\author{A. Ghosh}
 \address{Mathematics and Applied Mathematics, MAI, Link\"oping University, SE 58183 Linköping, Sweden}
 \email{arpan.ghosh@liu.se}

\author{V. A. Kozlov}
 \address{Mathematics and Applied Mathematics, MAI, Link\"oping University, SE 58183 Linköping, Sweden}
 \email{vladimir.kozlov@liu.se}

\author{S. A. Nazarov}\thanks{S. A. Nazarov acknowledges the support from Russian Foundation of Basic Research, grant 15-01-02175.}
 \address{St. Petersburg State University, 198504, Universitetsky pr., 28, Stary Peterhof, Russia; Peter the Great St. Petersburg State Polytechnical University, Polytechnicheskaya ul., 29, St. Petersburg, 195251, Russia; Institute of Problems of Mechanical Engineering RAS, V.O., Bolshoj pr., 61, St. Petersburg, 199178, Russia
}
 \email{srgnazarov@yahoo.co.uk}

\author{D. Rule}
 \address{Mathematics and Applied Mathematics, MAI, Link\"oping University, SE 58183 Linköping, Sweden}
 \email{david.rule@liu.se}

\begin{abstract}
We present a two dimensional model describing the elastic behaviour of the wall of a curved blood vessel. The wall has a laminate structure consisting of several anisotropic layers of varying thickness and is assumed to be much smaller in thickness than the radius of the vessel which itself is allowed to vary. Our two-dimensional model takes the interaction of the wall with the surrounding tissue and the blood flow into account and is obtained via a dimension reduction procedure. The curvature and twist of the vessel axis as well as the anisotropy of the laminate wall present the main challenges in applying the dimension reduction procedure so plenty of examples of canonical shapes of vessels and their walls are supplied with explicit systems of differential equations at the end.
\end{abstract}

\maketitle

%

\section{Introduction}
\label{Intro}
The circulatory system is one of the most important systems in the human body as it serves a number of functions such as supplying nutrition throughout the body, regulating temperature and fighting diseases. It is also a system that is susceptible to various kinds of risks. Having an accurate model for the system can go a long way in helping early diagnosis and devising appropriate management strategies for any arising problem. A lot of effort has been put into the modelling of blood flow through blood vessels, see, for example, the monograph \cite{Fung}. The complexity of the arrangement of elastic tissues forming a blood vessel adds to the difficulty in accurately modeling the interaction of blood flow with the vessel wall. The vessel walls have a laminate structure consisting of three layers of tissues (called adventitia, media and intima) having different composition and elastic properties, see \cite{Fratzl}. 

There are a few ways to model the elastic properties of the vessel walls of vessels having a cylindrical reference geometry by introducing some simplifying assumptions on the wall structure. For example, assuming a thin shell model for the wall results in the models presented in \cite{QuFo04}. Another shell model called the Koiter shell model along with the Kelvin-Voigt model for viscoelastic materials has been used in \cite{CTGMHR}. Navier equations are also used in modeling the elastic walls by assuming the vessel wall to be an elastic membrane, see \cite{Fung, QTV2000, CaMi}. However, all these models, and even more elementary versions of the model presented here \cite{KoNa11, KoNa16}, only deal with walls of straight sections of arteries despite the obvious existence of their curvilinear counterparts in the circulatory system. The curvilinearity of the vessels also occurs as a result of movements of parts of the body. Hence, better simulations of real blood vessels should be achieved when their curvilinearity is taken into account in the model.

Our aim with this article is to derive a two dimensional model describing the elastic behaviour of the wall of a curved blood vessel (allowing arbitrarily high curvature and torsion). Apart from the general geometry considered in this article, the novelty of our work lies in the fact that we derive the model taking into account the laminate structure of the wall as well as the anisotropic structure of each layer. We have also included the mechanical influence of the tissues surrounding the vessel as this can also play a vital role in blood flow, see \cite{MXAFCTG}. Furthermore, we have an extra term in our model representing other external influences such as those produced by the movement of limbs. Assuming that the thickness of the wall is small compared to the diameter of the blood flow channel, we perform dimension reduction following the classical scheme of asymptotic analysis, cf. \cite{Shoi, Ciar, CiDe, NaSwSl} and others, but we modify it and, as is important to ensure the calculations are tractable, adapt the convenient Voigt-Mandel notation to the curvilinear non-orthogonal coordinate system used. We build on the work for the case of a straight cylinder in \cite{KoNa11, KoNa16} and develop the model for a very general vessel geometry. We provide examples of some simple cases with explicit equations for these cases in (\ref{Ex1}), (\ref{Ex2}), (\ref{Ex3}), (\ref{Ex4}) and (\ref{Ex5}).

\subsection{Formulation of the problem}\label{formprob}
\begin{figure}
	\centering
	\begin{subfigure}{.4\textwidth}
		\centering
		\includegraphics[scale=0.5]{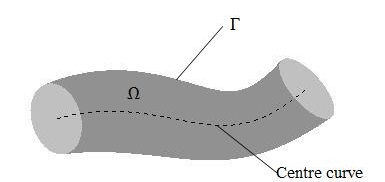}
		\label{Fig1}
	\end{subfigure}
	\begin{subfigure}{.4\textwidth}
		\centering
		\includegraphics[scale=0.3]{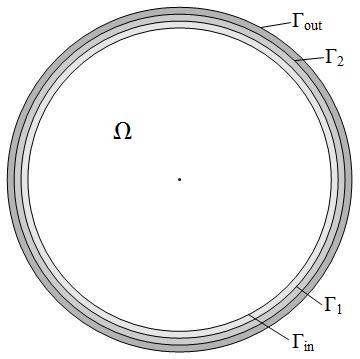}
		\label{Fig2}
	\end{subfigure}
	\caption{A section of a vessel along with a magnified cross-section showing the layered structure of the vessel wall. The region $\G$ constitutes of $\Gout,\G_1,\G_2,\Gin$ and the space enclosed between these layers. The hollow interior of the vessel is denoted by $\Om$.}
	\label{fig:Vessel}
\end{figure}

We consider a segment of a blood vessel and denote the hollow interior, which is also known as lumen, by $\Om$, as shown in figure \ref{fig:Vessel}. Let $\G$ denote the part of the blood vessel wall surrounding the region $\Om$ except at the two open ends. In other words, $\G$ is a deformed hollow cylindrical pipe enveloping the region $\Om$. $\G$ has a layered structure with the layers separated by deformed cylindrical surfaces $\Gin,\G_1,\ldots,\G_m,\Gout$ where $\Gin$ denotes the interior boundary of $\G$ closest to $\Om$ and hence in direct contact with the blood flow and $\Gout$ denotes the exterior boundary of $\G$ that is adjacent to the surrounding muscle tissue. For a blood vessel, these layers are typically made of anisotropic elastic materials such as collagen or smooth elastic tissues which contribute to their elastic properties. We assume that a central curve through the interior of the vessel is given and has a general geometry allowing non-zero curvature and torsion. We use this central curve as a reference to describe the surrounding regions $\Om$ and $\G$ in suitable coordinates. We also assume that along the given central curve, the vessel has a circular cross-section with slightly varying radius along the length of the vessel. This assumption is based on a result in \cite{KoNa16} where an optimal property of the cross-section was verified. It should be noted here that the calculations are largely the same even for non circular cross-section, although the resulting expressions are more complicated.

We assume that we have a fixed Cartesian coordinate system for the ambient three dimensional space and we use it to describe other suitable curvilinear coordinate systems.

In order to formulate the problem at hand, we introduce a few notations. For some time interval $[0,T]$, let the velocity field of the blood flowing through the lumen of the vessel be denoted by $\bv:\Omt\to\Rt$. This provides us the velocity of the blood particles at any given location in $\Om$ at any particular instance in time. The displacement field in the vessel wall is denoted by $\bu:\Gt\to\Rt$, $\bu=(u_1,u_2,u_3)$. This describes the displacement of any point of the vessel wall after deformation with respect to its position in the undeformed vessel wall. Let $p:\Omt\to\R$ denote the ratio of pressure within the blood in the vessel and the blood density which is assumed to be constant in our case.

We further assume that the deformations in the vessel wall are sufficiently small, allowing us to employ the linearized theory of elasticity. In order to describe the stresses present in the material, we use the Cauchy stress tensor, denoted by $\s = \{\s_{ij}\}_{i,j=1}^3$. On the other hand, we use the linear strain tensor, denoted by $\e = \{\e_{ij}\}_{i,j=1}^3$, to quantify the infinitesimal deformations in the vessel wall. See \cite{Lurie, LLTOE} for a detailed description of the stress and strain tensors.

Now we are in a position to introduce the governing equations of elasticity theory in our case of the elastic vessel wall. We start with the relation between the stress tensor and the linear strain tensor in the vessel wall. According to Hooke's law for linear elasticity, the stress tensor is linearly dependent on the strain tensor. Hence,
\begin{equation}
	\s_{ij} = \sum\limits_{k,l=1}^3{A_{ij}^{kl}\e_{kl}}
\end{equation}
with $A_{ij}^{kl}$ being the components of the stiffness tensor $\mathcal{A}$ having the symmetries $A_{ij}^{kl}=A_{ji}^{kl}=A_{ji}^{lk}$. It is also coercive so that $\sum\limits_{i,j,k,l=1}^3{A_{ij}^{kl}\xi_{kl}\xi_{ij}}\geq C_{\mathcal{A}}\sum\limits_{k,l=1}^3{|\xi_{kl}|^2}$ for some constant $C_{\mathcal{A}}>0$ and any rank 2  symmetric tensor $\{\xi_{kl}\}_{k,l=1}^3$. It is the stiffness tensor that contains the information about the elastic qualities of the material in question. We assume that $A_{ij}^{kl}$ are constant across a layer, although they can be different for different layers.

The strain tensor and the displacement vector have the relation 
\begin{equation}\label{Cauchy}
	\e_{kl} = \frac{1}{2}\left(\frac{\pd u_k}{\pd x_l}+\frac{\pd u_l}{\pd x_k}\right).
\end{equation}

Newton's second law of motion gives us the final set of equations for the vessel wall.
\begin{equation}\label{2ndlaw}
	\na\cdot\s=\rho\frac{\pd^2\bu}{\pd t^2}\mbox{ in }\G,
\end{equation}
 where $\rho$ is the vessel wall mass density which is assumed to be piecewise continuous across the layers.

We supplement the stated set of equations with the following boundary conditions. On the inner boundary $\Gin$, we have an equilibrium condition due to traction in the wall, given as $\s\bn$ for the outward unit normal $\bn$, balancing out the hydrodynamic force from the fluid motion. We also have a dynamic no-slip condition equating the velocities of the wall surface and that of the fluid at the surface. So, with $h$ being a small parameter denoting the ratio of the average thickness of the wall to a chosen reference radius of the vessel, we have
\begin{equation}\label{ibc}
	\s\bn=h\rho_b\bF\mbox{ and }\pd_t\bu=\bv\mbox{ on }\Gin,
\end{equation}
where $\rho_b$ is the blood density and $\bF$ is the normalized\footnote{With the introduction of the small parameter $h$ representing the thinness of the wall, the forces acting on the wall need to be appropriately adjusted in accordance with $h$ in order to get reasonable orders of magnitudes that facilitate correct asymptotic analysis. One can expect for example that a very thin wall can withstand only forces having small orders of magnitude. Hence we normalize the forces to have the factor $h$ with them.} hydrodynamic force in the blood given by
\[\bF=-p\bn+2\nu\Def(\bv)\bn\]
where $\Def$ is the symmetrized gradient operator and $\nu$ is the dynamic viscosity of blood.

On the outer boundary $\Gout$, we again have a balance of forces exerted by the surrounding muscle material, external forces and traction. Hence, again with the same small parameter $h$, we get
\begin{equation}\label{obc}
\s\bn+h\mathcal{K}\bu=h\mathbf{f}\mbox{ on }\Gout,
\end{equation}
where $h\mathcal{K}$ is the tensor corresponding to the elastic response of the surrounding muscle tissue so that $\mathcal{K}\bu=k(\bu\cdot\bn)\bn$ for some given constant $k$ and $\mathbf{f}$ is the normalized force exerted on the vessel by external factors. In most cases, $\mathbf{f}=\mathbf{0}$ as the effect of external factors are negligible compared to the forces exerted by the surrounding muscle material. The tensor $\mathcal{K}$ is described in  the Appendix \ref{KTensor}.

\section{Geometric setup and notations}
\label{PreNot}

\subsection{Setting up a curvilinear coordinate system}

In many modeling problems, the choice of a coordinate system proves to be crucial as it greatly affects the ease with which we can carry out computations. Having this in mind, we begin modeling by choosing a suitable coordinate system that simplifies the computations even in the case of the most general geometry of the vessel. We assume a centre curve of the vessel to be known and given by an arc-length perameterized  curve $\bc\in\mathcal{C}^2([0,L],\Rt)$ for some positive real $L$ that represents the total length of the considered vessel. We may assume the initial conditions 
\begin{equation}
\bc(0)=(0,0,0)^T\mbox{ and }\bc'(0)=(0,0,1)^T.
\label{curvinit}
\end{equation}
Henceforth, if a function, say $f$, depends only on one variable, we denote its derivative by $f'$. The arc-length parameter is denoted by $s$. We use this centre curve to develop the required coordinate frames.

At first, we need to build a right handed coordinate frame at each point $\bc(s)$. It is natural to take one of the coordinate directions to be $\bc'(s)$ and the other two to be perpendicular to it. Let $\be_1$ be one such unit vector perpendicular to $\bc'$ at each $s$.

We could use the Frenet frame to define $\be_1(\theta,s)=\cos{\theta}\mathbf{N}(s)-\sin{\theta}\mathbf{B}(s)$ where $\mathbf{N}$ and $\mathbf{B}$ are the unit normal and the unit binormal of the curve $\bc$, if we assume that the curve has non vanishing curvature. Consider the surface $S(\theta,s)=\bc(s)+r_\delta\be_1(\theta,s)$ for some $r_\delta>0$. The surface $S$ is a pipe around the centre curve so that every point on it has a constant distance $r_\delta$ from the centre curve. Simple calculations with the help of the Serret-Frenet formulas show that, $r_\delta^{-2}(\pd S(\theta,s)/\pd\theta)\cdot(\pd S(\theta,s)/\pd s)$ is equal to the negative of the torsion of the curve $\bc$. As the partial derivatives of $S$ with respect to $\theta$ and $s$ give us the coordinate directions for the respective parameters, we conclude that the coordinate lines do not intersect at right angles when $\bc$ has nonzero torsion. For this reason, we reject this choice of coordinate frame.

A better choice proves to be one based on the requirement that the change in $\be_1$ as we travel along the central curve, should be coplanar with $\be_1$ and $\bc'(s)$, that is, $\pd_s\be_1\cdot(\be_1\times\bc')=0$. This prevents the coordinate lines corresponding to $s$ from `wrapping around' the tubular surface due to torsion of the centre curve. Since $\be_1$ is a unit vector, we have
\[\pd_s\be_1\cdot\be_1=0.\]
Also, as $\bc'$ is perpendicular to $\be_1$ at each $s$, it follows that
\[\pd_s\be_1\cdot\bc'=-\bc''\cdot\be_1.\]
Therefore, the coplanarity condition on the vector $\pd_s\be_1$ with the orthonormal vectors $\{\bc',\be_1\}$ is equivalent to
\[\pd_s\be_1=(\pd_s\be_1\cdot\bc')\bc'+(\pd_s\be_1\cdot\be_1)\be_1=-(\bc''\cdot\be_1)\bc'.\]
We choose the initial value of $\be_1$ at $s=0$ to be $(\cos{\theta},\sin{\theta},0)^T$ for some $\theta\in[0,2\pi]$. Then we obtain the following initial value problem that defines $\be_1$
\begin{equation}\label{e1}
	\pd_s\be_1(\theta,s)=-(\bc''(s)\cdot\be_1(\theta,s))\bc'(s)\mbox{ and }\be_1(\theta,0)=(\cos{\theta},\sin{\theta},0)^T.
\end{equation}

Defining $\be_2(\theta,s)=\bc'(s)\times\be_1(\theta,s)$, the triple $\{\be_1(\theta,s),\be_2(\theta,s),\bc'\}$ forms an orthonormal frame at each point $\bc(s)$ for a given angle $\theta$. As a result, we have 
\begin{align*}
	\pd_s\be_2(\theta,s)&=\bc''(s)\times\be_1(\theta,s)+\bc'(s)\times\pd_s\be_1(\theta,s)=\bc''(s)\times(\be_2(\theta,s)\times\bc'(s))+0\\
																	&=\bc''(s)\cdot\bc'(s)\be_2(\theta,s)-\bc''(s)\cdot\be_2(\theta,s)\bc'(s)=-\bc''(s)\cdot\be_2(\theta,s)\bc'(s).
\end{align*}
The initial condition for $\be_2$ reads \[\be_2(\theta,0)=\bc'(0)\times\be_1(\theta,0)=(-\sin{\theta},\cos{\theta},0)^T.\]
So we obtain
\begin{equation}\label{e2}
	\pd_s\be_2(\theta,s)=-(\bc''(s)\cdot\be_2(\theta,s))\bc'(s)\mbox{ and }\be_2(\theta,0)=(-\sin{\theta},\cos{\theta},0)^T.
\end{equation}

The equations \eqref{e1} and \eqref{e2} ensure that the frame $\{\bc',\be_1,\be_2\}$ is a so called `rotation minimizing frame', see \cite{Bish,Klok}.

One can find a rotation-matrix valued function $R$ so that $\be_i(\theta,s)=R(s)\be_i(\theta,0)$ for $i=1,2$. Then it is readily obtained that \[\pd_\theta\be_1(\theta,s)=\be_2(\theta,s)\mbox{ and }\pd_\theta\be_2(\theta,s)=-\be_1(\theta,s).\]

The parameter $\theta$ corresponds to the orientation of the vectors $\be_1(\theta,s)$ and $\be_2(\theta,s)$ for given $s$ with respect to some reference vector in the same disc perpendicular to the corresponding tangent vector $\bc'(s)$ of the central curve. Note that in the torsion free case, when $\bc''$ is never zero and $\bc''(0)=(1,0,0)^T$, the orthonormal frame is same as $\{\cos{\theta}\mathbf{N}(s)-\sin{\theta}\mathbf{B}(s),\sin{\theta}\mathbf{N}(s)+\cos{\theta}\mathbf{B}(s),\bc'(s)\}$ where $\mathbf{N}$ and $\mathbf{B}$ are respectively the unit normal and the unit binormal of the curve $\bc$.

We have two parameters, namely, $\theta$ and $s$, that describe the inner boundary of the wall of the vessel. Next we construct a coordinate system in the wall. In order to include the information of the layered structure of the vessel wall, we assume that the layers are given as level sets of a sufficiently smooth function $G:\Rt\rightarrow\R$. The innermost layer is given as $\{\bx\in\Rt|G(\bx)=0\}$ while the outermost layer is given as $\{\bx\in\Rt|G(\bx)=H\}$ where $H>0$ is some fixed reference thickness of the vessel wall. In particular, $G$ could be assumed to have the form $G(\bx)=d(\bx)a(\bx)$ where $d(\bx)$ is the distance of $\bx$ from $\Gin$ and $a(\bx)$ is a suitable scaling so as to keep $G$ constant over a given surface. The normal vector field across the layers is given by $\na G$. Let another parameter $n$ be such that it corresponds to the layer (we assume there exists a continuum of level surfaces filling up $\G$) to which a given point belongs. In other words, let $n=G(\bx)$. Differentiating with respect to $n$, we get \[1=\na G(\bx)\cdot\pd_n\bx.\] On the other hand, the integral curves of the vector field $\pd_n\bx$, have tangents $\pd_n\bx$ parallel to $\na G$. Hence, with the help of the relation above, we get the integral curves by solving ordinary differential equation \[\pd_n\bx(n,\theta,s)=\frac{\na G(\bx(n,\theta,s))}{|\na G(\bx(n,\theta,s))|^2}.\]
The initial condition on such lines are that they originate on $\Gin$ where $n=0$. In other words, for some $\theta\in[0,2\pi]$ and $s\in [0,L]$,\[\bx(0,\theta,s)=\bc(s)+r(s)\be_1(\theta,s)\] where $r(s)$ is the given radius of the interior channel.

Thus, we have three parameters describing the vessel wall in a curvilinear coordinate system. The parameter $n$ corresponds to the direction perpendicular to the layers, $s$ corresponds to the tangential direction along the central curve while $\theta$ corresponds to the direction tangential to the closed curve determined by fixed $n$ and $s$.

The relation between the Cartesian and the curvilinear coordinate system in the wall is \[\bx(n,\theta,s)=\bc(s)+r(s)\be_1(\theta,s)+\int^{n}_{0}\frac{\na G(\bx(\tau,\theta,s))}{|\na G(\bx(\tau,\theta,s))|^2}\dif \tau.\]

\subsection{Basis vectors and differential operators}

In order to use the formulae mentioned in Subsection \ref{formprob} in terms of the new coordinates, we need to express the vectors, tensors and differential operators in a suitable basis. A detailed presentation of tensor algebra in curvilinear coordinates for application to continuum mechanics can be found in Appendix D of \cite{Lurie}.

In what follows, we let $\pd_1=\pd/\pd n$, $\pd_2=\pd/\pd\theta$ and $\pd_3=\pd/\pd s$. Also, we adopt Einstein's summation convention, that is, repeated indices (when appearing concurrently at both top and bottom positions in a term) are assumed to be summed over the index set, which is $\{1,2,3\}$ in our case.

We now define a set of contravariant basis vectors for tangent vectors inside the wall structure. Let $\bx_i=\pd_i\bx$ for $i=1,2,3$ and some $\bx\in\G$. This leads to the definition of the rank $2$ metric tensor as $g_{ij}=\bx_i\cdot\bx_j$ for $i,j=1,2,3$. Let $g$ denote the matrix $[g_{ij}]$.

We may also define a set of covariant or reciprocal basis vectors for the same space (cf. Appendix D of \cite{Lurie}) which are given as $\bx^i={g^{ij}\bx_j}$ where $g^{ij}$ is such that ${g^{ij}g_{jk}}=\delta_k^i$ with $\delta_k^i$ being the Kronecker delta. Note that contra-basis vectors have bottom indices while the reciprocal basis vectors have top indices. In our case, both the vectors $\bx_1$ and $\bx^1$ are parallel to the normal direction across the layers at each point in the wall. This makes it easier for us to formulate the physical laws.

In order to express derivatives in a curvilinear system, we need the Christoffel symbols corresponding to the curvilinear system which are defined as $\G^i_{jk}=\bx^i\cdot\pd_j\bx_k$ for $i,j,k=1,2,3$. They are symmetric in the lower indices, i.e., $\G^i_{jk}=\G^i_{kj}$.

With the help of these relations, we can define the gradient operator as $\na=\bx^i\pd_i$, see appendix E in \cite{Lurie}. We are now in a position to express quantities like gradient and divergence of tensors in our curvilinear coordinates. For any vector $\bv=v_j\bx^j$, its gradient is given as \[\na\bv=\left(\pd_iv_j-\G^k_{ij}v_k\right)\bx^i\bx^j.\] Also, for any rank $2$ tensor $\s=\s^{ij}\bx_i\bx_j$, its divergence is expressed as \[\na\cdot\s=(\pd_i\s^{ik}+\G^i_{ij}\s^{jk}+\G^k_{ij}\s^{ij})\bx_k.\] Similarly, the deformation (symmetric gradient) tensor for any vector turns out to be \[2\Def(\bv)=(\pd_iv_j+\pd_jv_i-2\G^k_{ij}v_k)\bx^i\bx^j.\]

\subsection{Volume elements}

The infinitesimal volume element with respect to the new variables is
\begin{equation}
	\dif v = \sqrt{\Det(g)}\dif n\dif\theta\dif s,
\end{equation}
where $\Det(\cdot)$ denotes the determinant.

Similarly, the infinitesimal surface element on a surface with fixed $n$ is
\begin{equation}\label{Smeas}
	\dif S = \sqrt{g_{22}g_{33}-(g_{23})^2}\dif\theta\dif s.
\end{equation}

\subsection{Rearranged Voigt-Mandel notation}

As is evident from the formulae above, one has to deal with a good number of indices in each of the equations. The stress and strain tensors are each rank $2$ tensors whereas the stiffness tensor is a rank $4$ tensor which has $81$ components. However, owing to certain symmetries, the number of independent components are only $21$. Certain rotational symmetries further bring down the number to 18. We use a rearranged\footnote{We rearrange the terms in the standard Voigt-Mandel notation in order to simplify the presentation of the dimension reduction procedure.} Voigt-Mandel notation to write only the independent quantities. In this notation, the strain tensor $\e=\e_{ij}\bx^i\bx^j$ is represented as \[\bve=\left(\e_{11},\sqrt{2}\e_{12},\sqrt{2}\e_{13},\e_{22},\e_{33},\sqrt{2}\e_{23}\right)^T.\]

We express the stress tensor $\s=\s^{ij}\bx_i\bx_j$ in the Voigt-Mandel notation by \[\bsi=\left(\s^{11},\sqrt{2}\s^{12},\sqrt{2}\s^{13},\s^{22},\s^{33},\sqrt{2}\s^{23}\right)^T.\]

On the other hand, the symmetric gradient operator is represented as the matrix $D$ defined via 
\begin{equation}\label{defmat}
D^T=\left[
\begin{smallmatrix}
\pd_1-\G^1_{11}&2^{-1/2}(\pd_2-2\G^1_{12})&2^{-1/2}(\pd_3-2\G^1_{13})&		 -\G^1_{22}&		 -\G^1_{33}&2^{-1/2}(			-2\G^1_{23})\\
		 -\G^2_{11}&2^{-1/2}(\pd_1-2\G^2_{12})&2^{-1/2}(		 -2\G^2_{13})&\pd_2-\G^2_{22}&		 -\G^2_{33}&2^{-1/2}(\pd_3-2\G^2_{23})\\
		 -\G^3_{11}&2^{-1/2}(			-2\G^3_{12})&2^{-1/2}(\pd_1-2\G^3_{13})&		 -\G^3_{22}&\pd_3-\G^3_{33}&2^{-1/2}(\pd_2-2\G^3_{23})
\end{smallmatrix}\right],
\end{equation} so that for a vector $\bu=u_i\bx^i$, $D(u_1,u_2,u_3)^T$ gives the Mandel-Voigt notation for the rank $2$ tensor $\Def(\bu)$ in the chosen basis.

The divergence operator for a rank $2$ tensor is represented as the matrix $-D^*$ where $D^*$ is the Hermitian conjugate of $D$ with respect to the surface measure defined in (\ref{Smeas}), i.e. \[\int_{S}(D^*u)^Tv\sqrt{g_{22}g_{33}-(g_{23})^2}\dif\theta\dif s=\int_{S} u^T(Dv)\sqrt{g_{22}g_{33}-(g_{23})^2}\dif\theta\dif s\] for all $u,v\in L^2(\R^2,\R^3)$ such that $u|_{s=0,L}=0=v|_{s=0,L}$ and any level surface $S$ contained in $\G$. We have in this case,\[D^*=-\left[
\begin{smallmatrix}
\pd_1+\G^i_{i1}+\G^1_{11}&2^{-1/2}(\pd_2+\G^i_{i2}+2\G^1_{12})&2^{-1/2}(\pd_3+\G^i_{i3}+2\G^1_{13})&								\G^1_{22}&								\G^1_{33}&2^{-1/2}(								 2\G^1_{23})\\
								\G^2_{11}&2^{-1/2}(\pd_1+\G^i_{i1}+2\G^2_{12})&2^{-1/2}(								2\G^1_{13})&\pd_2+\G^i_{i2}+\G^2_{22}&								\G^2_{33}&2^{-1/2}(\pd_3+\G^i_{i3}+2\G^2_{23})\\
								\G^3_{11}&2^{-1/2}(								 2\G^3_{12})&2^{-1/2}(\pd_1+\G^i_{i1}+2\G^1_{13})&								\G^3_{22}&\pd_3+\G^i_{i3}+\G^3_{33}&2^{-1/2}(\pd_2+\G^i_{i2}+2\G^3_{23})
\end{smallmatrix}\right]
\] so that for $\s=\s^{ij}\bx_i\bx_j$, $-D^*\left(\s^{11},\sqrt{2}\s^{12},\sqrt{2}\s^{13},\s^{22},\s^{33},\sqrt{2}\s^{23}\right)^T$ is the coordinate matrix corresponding to the vector $\na\cdot\s$ expressed in the contravariant basis.

\section{Modelling of elastic walls}
\label{ModEW}

In this section, we shall obtain our two-dimensional model of the elastic vessel walls. We follow the steps in \cite{KoNa16}, i.e. we perform dimension reduction of the model given by equations (\ref{2ndlaw}),(\ref{ibc}) and (\ref{obc}) by identifying a small parameter and assuming asymptotic expansions of the displacement vector, the stress and the strain tensors.

\subsection{Asymptotic ansatz}

A property of the vessel walls considered here is that the thickness of the wall is small compared to some characteristic radius $r_0$. So a natural choice for a small parameter in our case is $h=H/r_0$. This means that the physical quantities in question change much faster across the wall layers as compared to along the layers. This prompts us to introduce a fast variable $\xi=h^{-1}n\in[0,r_0]$.

We then assume that the displacement vector $\bu$ admits the expansion 
\begin{equation}\label{Ans}
	\bu(n,\theta,s)=\bu_0(\xi,\theta,s)+h\bu_1(\xi,\theta,s)+h^2\bu_2(\xi,\theta,s)+\cdots.
\end{equation}
We denote the coordinate vector of $\bu_k$ in the basis $(\bx^1,\bx^2,\bx^3)$ by $U_k=(u_{k1},u_{k2},u_{k3})^T$.

The differential operator $D$ defined in (\ref{defmat}) can also be expanded as 
\begin{equation}
	D=h^{-1}B\pd_\xi+E+hD_1+h^2D_2+\cdots,
\end{equation}
due to the change of variable, where $E=C+D_0$ and
\begin{equation}
	B^T=\begin{bmatrix}
1&								 0&									0&0&0&0\\
0&2^{-1/2}&									0&0&0&0\\
0&								 0&2^{-1/2}&0&0&0
\end{bmatrix},\quad
	C^T=\begin{bmatrix}
	0&2^{-1/2}\pd_2&2^{-1/2}\pd_3&		0&		0&											0\\
	0&											0&											0&\pd_2&		0&2^{-1/2}\pd_3\\
	0&											0&											0&		0&\pd_3&2^{-1/2}\pd_2
\end{bmatrix}.
\end{equation}
With $(\G^k_{ij})_l$ denoting the coefficient of $h^l$ in the infinite series expression of $\G^k_{ij}$ (see Appendix \ref{AppExp}), we have for $m\geq 0$
\begin{equation}
	D_{m}^T=-\begin{bmatrix}
(\G^1_{11})_m&\sqrt{2}(\G^1_{12})_m&\sqrt{2}(\G^1_{13})_m&(\G^1_{22})_m&(\G^1_{33})_m&\sqrt{2}(\G^1_{23})_m\\
(\G^2_{11})_m&\sqrt{2}(\G^2_{12})_m&\sqrt{2}(\G^2_{13})_m&(\G^2_{22})_m&(\G^2_{33})_m&\sqrt{2}(\G^2_{23})_m\\
(\G^3_{11})_m&\sqrt{2}(\G^3_{12})_m&\sqrt{2}(\G^3_{13})_m&(\G^3_{22})_m&(\G^3_{33})_m&\sqrt{2}(\G^3_{23})_m
\end{bmatrix}.
\end{equation}
We also use the Taylor series of $\na G$:
\begin{align*}
	\na G(\bx(h\xi,\theta,s))&=\na G(\bx(0,\theta,s))+h\xi\pd_n\na G(\bx(0,\theta,s))+\frac{h^2\xi^2}{2}\pd^2_n\na G(\bx(0,\theta,s))+\cdots.
\end{align*}
For the distance function $d$, we have that $\|\na d(\bx)\|=1$ for all $\bx\in\G$. Also, $d(\bx(0,\theta,s))=0$. As we have $G=da$, hence,
\[|\na G(\bx(0,\theta,s))|=|a(\bx(0,\theta,s))\na d(\bx(0,\theta,s))+d(\bx(0,\theta,s))\na a(\bx(0,\theta,s))|=|a(\bx(0,\theta,s))|.\]

\subsection{The two dimensional model}

Let $F$ denote the coordinate vector of the hydrodynamic force $\bF$ in the basis $(\bx^1,\bx^2,\bx^3)$. Let $M$ be the leading term in the Taylor series expansion of $g^{-1}$ with respect to $h$. So we have
\[
	M=\begin{bmatrix}
	|a(\bx(0,\theta,s))|^2&0&0\\
						0&r^{-2}&0\\
						0&0&[(1-r\bc''\cdot\be_1)^2+r'^2]^{-1}
\end{bmatrix}.
\]
Furthermore, assume $E^T=[E_1^T|E_2^T]$ with each block being a $3\times 3$ matrix. Hence,
\[E_2=\begin{bmatrix}
			 -(\G^1_{22})_0&													 \pd_2-(\G^2_{22})_0&-(\G^3_{22})_0\\
			 -(\G^1_{33})_0&																-(\G^2_{33})_0&\pd_3-(\G^3_{33})_0\\
-2^{1/2}(\G^1_{23})_0&								2^{-1/2}(\pd_3-2(\G^2_{23})_0)&2^{-1/2}(\pd_2-2(\G^3_{23})_0)
\end{bmatrix}.\]
Similarly, we have 
\[E_2^*=-\begin{bmatrix}
										(\G^1_{22})_0&										(\G^1_{33})_0&2^{1/2}(											 \G^1_{23})_0\\
\pd_2+(\G^i_{i2})_0+(\G^2_{22})_0&										(\G^2_{33})_0&2^{-1/2}(\pd_3+(\G^i_{i3})_0+2(\G^2_{23})_0)\\
										(\G^3_{22})_0&\pd_3+(\G^i_{i3})_0+(\G^3_{33})_0&2^{-1/2}(\pd_2+(\G^i_{i2})_0+2(\G^3_{23})_0)
\end{bmatrix}.\]

By $A=\begin{bmatrix}
	 A_{\dagger\dagger}& A_{\dagger\ddagger}\\
A_{\dagger\ddagger}^T&A_{\ddagger\ddagger}
\end{bmatrix}$, we denote the $6\times 6$ matrix, with each block being a $3\times 3$ matrix, corresponding to the stiffness tensor such that 
\[
	A\left(\e_{11},\sqrt{2}\e_{12},\sqrt{2}\e_{13},\e_{22},\e_{33},\sqrt{2}\e_{23}\right)^T=\left(\s^{11},\sqrt{2}\s^{12},\sqrt{2}\s^{13},\s^{22},\s^{33},\sqrt{2}\s^{23}\right)^T.
\]
We set $K$ to be the matrix representation of the tensor $\mathcal{K}$ in the appropriate basis. In our case, we obtain
\begin{equation}\label{k00}
		K=\begin{bmatrix}
	k&0&0\\
	0&0&0\\
	0&0&0
	\end{bmatrix}.
\end{equation}

With the above notations, we have the following theorem that gives us a two dimensional model of the wall of a vessel.
\begin{thm}\label{thm}
In the asymptotic expansion\footnote{We assume (\ref{Ans}) to be an asymptotic expansion which can be shown mathematically to be true but we omit it as it is not relevant for the primary focus of this article.} of the displacement vector $\bu$ given in (\ref{Ans}), we have that $U_0$ is independent of $\xi$ and it satisfies the relation
\begin{equation}\label{model}
	E_2^*QE_2U_0+M(\bar{\rho}\pd_t^2U_0+|\na G(\bx(0,\theta,s))|KU_0)=|\na G(\bx(0,\theta,s))|M(F_{ext}-\rho_bF),
\end{equation}
where \[Q=\int\limits_0^{r_0}(A_{\ddagger\ddagger}-A_{\dagger\ddagger}^TA_{\dagger\dagger}^{-1}A_{\dagger\ddagger})\dif\xi,\quad\bar{\rho}=\int\limits_0^{r_0}\rho\dif\xi\] and $F_{ext}$ denotes the column representing $\mathbf{f}$ in the basis $(\bx^1,\bx^2,\bx^3)$.
\end{thm}
\begin{proof}
Choosing the basis $(\bx^1,\bx^2,\bx^3)$ to express the vectors, equation (\ref{2ndlaw}) can be written as 
\begin{equation}\label{source}
	-D^*ADU=g^{-1}\rho\pd_t^2U\mbox{ in }\G,
\end{equation}
where $U=[u_1,u_2,u_3]^T$ so that $\bu=u_i\bx^i$. 

Noting that the unit normal across the layers is $(g_{11})^{1/2}\bx^1$, we have that the outer boundary condition is 
\begin{equation}\label{sourcebc}
	B^TADU=h(g^{11})^{1/2}g^{-1}(F_{ext}-KU)\mbox{ on }\Gout,
\end{equation}
while the inner boundary condition results in
\begin{equation}\label{sourcebcin}
	B^TADU=h\rho_b(g^{11})^{1/2}g^{-1}F\mbox{ on }\Gin.
\end{equation}

After applying the substitution $\xi=h^{-1}n$ along with the asymptotic ansatz (\ref{Ans}) to (\ref{source}) and (\ref{sourcebc}), we compare the terms of the same orders of $h$ on both sides of the resulting equation. Comparing the coefficients of $h^{-2}$ in (\ref{source}) and those of $h^{-1}$ in (\ref{sourcebc}), we get the following system
\[
	B^T\pd_\xi AB\pd_\xi U_0=\hnt\mbox{ in }\G,
\]
\[
	B^TAB\pd_\xi U_0=\hnt\mbox{ on }\Gout.
\]
Solving this system and using the fact that $B^TAB$ is a non-singular $3\times3$ matrix, we obtain 
\begin{equation}
\pd_\xi U_0=\hnt\mbox{ in }\G.
\label{Uind}
\end{equation}
This proves that $U_0$ is independent of $\xi$.

Next we compare the coefficients of $h^{-1}$ in (\ref{source}) and those of $h^{0}$ in (\ref{sourcebc}). Using (\ref{Uind}), we get
\[
	B^T\pd_\xi A(B\pd_\xi U_1+EU_0)=\hnt\mbox{ in }\G,
\]
\[
	B^TA(B\pd_\xi U_1+EU_0)=\hnt\mbox{ on }\Gout.
\]
Solving the above system, we get
\[
	B^TA(B\pd_\xi U_1+EU_0)=\hnt
\]
\begin{equation}\label{u1}
	\Leftrightarrow\quad\pd_\xi U_1=-(B^TAB)^{-1}B^TAEU_0.
\end{equation}

Lastly, we compare the coefficients of order $h^{0}$ in (\ref{source}) and $h$ in (\ref{sourcebc}). This leads us to the following system:
\begin{equation}\label{lastorder}
	B^T\pd_\xi A(B\pd_\xi U_2+EU_1+D_1U_0)-E^*A(B\pd_\xi U_1+EU_0)=M\rho\pd_t^2U_0\mbox{ in }\G,
\end{equation}
\begin{equation}\label{lastorderBC1}
	B^TA(B\pd_\xi U_2+EU_1+D_1U_0)=|\na G(\bx(0,\theta,s))|M(F_{ext}-KU_0)\mbox{ on }\Gout.
\end{equation}
On the other hand, the inner boundary conditions (\ref{sourcebcin}) yield the relation
\begin{equation}\label{lastorderBC2}
	B^TA(B\pd_\xi U_2+EU_1+D_1U_0)=\rho_b|\na G(\bx(0,\theta,s))|MF\mbox{ on }\Gin.
\end{equation}

Note that $\Gout$ corresponds to $\xi=r_0=H/h$ while $\Gin$ corresponds to $\xi=0$. Integrating (\ref{lastorder}) with respect to $\xi$ from $0$ to $r_0$ and using (\ref{lastorderBC1}) and (\ref{lastorderBC2}), we have
\begin{equation}\label{laststep}
	|\na G_0(\bx(0,\theta,s))|M(F_{ext}-KU_0-\rho_bF)-\int\limits_0^{r_0}{E^*A(B\pd_\xi U_1+EU_0)\dif\xi}=M\bar{\rho}\pd_t^2U_0.
\end{equation}
Now (\ref{u1}) gives us that 
\[
	E^*A(B\pd_\xi U_1+EU_0)=E^*(A-AB(B^TAB)^{-1}B^TA)EU_0=E_2^*(A_{\ddagger\ddagger}-A_{\dagger\ddagger}^TA_{\dagger\dagger}^{-1}A_{\dagger\ddagger})E_2U_0.
\]
Integrating the above equation, we arrive at
\[
	\int\limits_0^{r_0}{E^*A(B\pd_\xi U_1+EU_0)\dif\xi}=E_2^*\left(\int\limits_0^{r_0}(A_{\ddagger\ddagger}-A_{\dagger\ddagger}^TA_{\dagger\dagger}^{-1}A_{\dagger\ddagger})\dif\xi\right) E_2U_0=E_2^*QE_2U_0.
\]
Using this relation in \eqref{laststep}, we have
\[
	E_2^*QE_2U_0+M(\bar{\rho}\pd_t^2U_0+|\na G(\bx(0,\theta,s))|KU_0)=|\na G(\bx(0,\theta,s))|M(F_{ext}-\rho_bF).\qedhere
\]

\end{proof}

\section{Examples of simpler cases}

In this section, we present a few simple cases and we look at the resulting expressions in the final model for each of these cases. We conclude each case by writing the explicit equations for the model. For this section, we use the notation $q_{ij}$ to denote the $ij$-th entry in the matrix $Q$ appearing in our model. Moreover, we take $q_{13}=0=q_{23}$ as is the case for orthotropic materials. Also, let $F_{ext,i}$ and $F_i$ denote the $i$th component of $F_{ext}$ and $F$ respectively.
The elastic properties of the wall are assumed to be uniform along the vessel length. Except case \ref{lastcase}, we let all the physical quantities such as forces and displacements have circular symmetry at each cross section and the vectors representing the physical quantities have zero component and variation along the circular direction. This leads to these cases having two equations each for the model as opposed to three equations for the case \ref{lastcase}. Note that all the functions in the folmulae in this section are evaluated at a point on $\Gin$, where $\xi=0$.

\subsection{Straight cylinder with uniform wall}

The simplest case for our model is that of a straight cylindrical vessel having constant radius and constant thickness for each layer of the wall. Here, $\bc''(s)=0$, $r'=0$ and $a\equiv 1$. Then with the initial conditions \eqref{curvinit} for the curve, we have for all $s\in[0,L]$
\[\bc(s)=(0,0,s)^T\quad\Rightarrow\quad\bc'(s)=(0,0,1)^T.\]
We get the orthonormal frame for all $s\in[0,L]$ and $\theta\in[0,2\pi]$ as
\[\be_1(\theta,s)=\be_1(\theta,0)=(\cos\theta,\sin\theta,0)^T\mbox{ and }\be_2(\theta,s)=\be_2(\theta,0)=(-\sin\theta,\cos\theta,0)^T.\]
For the distance function $d$ that measures distance from the innermost layer, we have
\[d(\bx)=\sqrt{x_1^2+x_2^2}-r\quad\Rightarrow\quad\na d(\bx)=\be_1(\theta,0),\]
where, $\bx=(x_1,x_2,x_3)^T$ and $\theta$ is such that $\cos\theta=x_1/\sqrt{x_1^2+x_2^2}$ and $\sin\theta=x_2/\sqrt{x_1^2+x_2^2}$.
Hence,
\[\na G(0,\theta,s)=\be_1(\theta,0).\]

The matrix $M$ is
\[
		M=\begin{bmatrix}
			1&						0&0\\
			0&r^{-2}&0\\
			0&						0&1
\end{bmatrix}.\]
The differential operator matrices $E_2$ and $E_2^*$ are given as
\[E_2=\begin{bmatrix}
 r&									 \pd_2&0\\
 0&											 0&\pd_3\\
 0&2^{-1/2}\pd_3&2^{-1/2}\pd_2
\end{bmatrix}
\mbox{ and }
E_2^*=-\begin{bmatrix}
	 -r&		0&0\\
\pd_2&		0&2^{-1/2}\pd_3\\
		0&\pd_3&2^{-1/2}\pd_2
\end{bmatrix}.\]
With all our assumptions, we get the following equations for \eqref{model}:
\begin{equation}\label{Ex1}\begin{aligned}
	 q_{11}r^2	 u_{01}+q_{12}r\pd_3	u_{03}+\bar{\rho}\pd_t^2u_{01}+ku_{01}=F_{ext,1}-\rho_bF_1,\\
	-q_{12}r\pd_3u_{01}- q_{22}\pd_3^2u_{03}+\bar{\rho}\pd_t^2u_{03}				=F_{ext,3}-\rho_bF_3.\end{aligned}
\end{equation}

\subsection{The straight cylinder with wall having variable thickness}

In this case, we assume once again that the central curve is a straight line. That is, $\bc''(s)=0$. We also take a fixed radius, so $r'=0$. Therefore we get the same expressions for $\bc$, $\be_1$, $\be_2$ and $d$ as in the previous case.
So we have,
\[\na G(\bx(0,\theta,s))=a(\bx(0,\theta,s))\be_1(\theta,0).\]

The matrix $M$ takes the form
\[
		M=\begin{bmatrix}
	|a|^2&						0&0\\
			0&r^{-2}&0\\
			0&						0&1
\end{bmatrix}.\]
The differential operator matrices $E_2$ and $E_2^*$ are
\[E_2=\begin{bmatrix}
ar&									 \pd_2&0\\
 0&											 0&\pd_3\\
 0&2^{-1/2}\pd_3&2^{-1/2}\pd_2
\end{bmatrix}
\] and \[
E_2^*=-\begin{bmatrix}
	-ar&												 0&0\\
\pd_2&												 0&2^{-1/2}(\pd_3-a^{-1}\na a\cdot\bc')\\
		0&\pd_3-a^{-1}\na a\cdot\bc'&2^{-1/2}\pd_2
\end{bmatrix}.\]
Finally, the system of equations \eqref{model} for this case are
\begin{equation}\label{Ex2}\begin{aligned}
&ar(q_{11}aru_{01}+q_{12}\pd_3	u_{03})+|a|^2(\bar{\rho}\pd_t^2u_{01}+|a|ku_{01})=|a|^3(F_{ext,1}-\rho_bF_1),\\
&(-\pd_3+a^{-1}\na a\cdot\bc')(q_{12}aru_{01}+q_{22}\pd_3u_{03})+\bar{\rho}\pd_t^2u_{03}					 =|a|(F_{ext,3}-\rho_bF_3).
\end{aligned}
\end{equation}

\subsection{Pipe with straight axis and equally spaced layers}
In this case as well, we assume that the central curve is a straight line. That is, $\bc''(s)=0$. So once again we have the same expressions for $\bc$, $\be_1$ and $\be_2$ as in the previous cases. The radius in this case is taken to be a function of the variable $s$.

The function $G$ has a simpler expression for this case as $a(\bx)=1$, owing to the fact that the layers are equally spaced. Hence,
\[G(\bx)=d(\bx)\quad\Rightarrow\quad|\na G|=1.\]
Let us use the notation $\gamma=(1+r'^2)^{-1/2}$. Then, we have
\[
	M=\begin{bmatrix}
	1&		 0&0\\
	0&r^{-2}&0\\
	0&		 0&\gamma^2
\end{bmatrix}.\]
The differential operator matrices $E_2$ and $E_2^*$ become
\[E_2=\begin{bmatrix}
	 r\gamma&										 \pd_2&rr'\gamma^2\\
-r''\gamma&												 0&\pd_3-r''r'\gamma^2\\
				 0&2^{-1/2}(\pd_3-2r^{-1}r')&2^{-1/2}\pd_2
\end{bmatrix}\]
and
\[E_2^*=-\begin{bmatrix}
		-r\gamma&										 r''\gamma&0\\
			 \pd_2&														 0&2^{-1/2}(\pd_3+r''r'\gamma^2+3r^{-1}r')\\
-rr'\gamma^2&\pd_3+r^{-1}r'+2r''r'\gamma^2&2^{-1/2}\pd_2
\end{bmatrix}.\]

The resulting equations for \ref{model} are
\begin{equation}\label{Ex3}\begin{aligned}
&\gamma[\lambda\pd_3u_{03}+\Lambda(u_{01}+r'\gamma u_{03})]+\bar{\rho}\pd_t^2u_{01}+ku_{01}=F_{ext,1}-\rho_bF_1,\\
&r'[(\gamma^2(\lambda-r''q_{22})-r^{-1}q_{22})\pd_3u_{03}+(\gamma^2(\Lambda-r''\lambda)-r^{-1}\lambda)(u_{01}+r'\gamma u_{03})]\\&-\pd_3(q_{22}\pd_3u_{03}+\lambda(u_{01}+r'\gamma u_{03}))+\gamma^2\bar{\rho}\pd_t^2u_{03}				=\gamma^2(F_{ext,3}-\rho_bF_3),\end{aligned}
\end{equation}
where $\lambda=rq_{12}-r''q_{22}$ and $\Lambda=\begin{bmatrix}r&-r''&0\end{bmatrix}Q\begin{bmatrix}r&-r''&0\end{bmatrix}^T.$

\subsection{Conical pipe with walls having proportionate thickness}
We consider a conical pipe with its central axis along the $x_3$-axis. Let the radius of the inner channel be $r(x_3)=mx_3+r_0$ for some initial radius $r_0$ and non-negative scalar $m$. We assume the thickness of the vessel wall to be proportional to the radius of the inner channel at each cross-section perpendicular to the central line. Then the functions $d$ and $a$ are
\[d(\bx)=(1+m^2)^{-1/2}(\sqrt{x_1^2+x_2^2}-mx_3-r_0)\quad\mbox{and}\quad a(\bx)=(mx_3+r_0)^{-1}r_0\sqrt{1+m^2}.\]
On $\Gin$, we have $s=x_3$ and we have
\[\na G(\bx(0,\theta,s))=(ms+r_0)^{-1}r_0(\cos{\theta},\sin{\theta},-m)^T.\]
The expressions for $\be_1$, $\be_2$ and $\theta$ are the same as the previous cases.

Hence, we arrive at the following expression for $M$;
\[M=\begin{bmatrix}
			(ms+r_0)^{-2}r_0^2(1+m^2)&						0&0\\
															0&(ms+r_0)^{-2}&0\\
															0&						0&(1+m^2)^{-2}
		\end{bmatrix}.\]
Next we have $E$ and $E^*$ as
\[E_2=\begin{bmatrix}
r_0&													\pd_2&(1+m^2)^{-1}(ms+r_0)m\\
	0&															0&\pd_3\\
	0&2^{-1/2}(\pd_3-2(ms+r_0)^{-1}m)&2^{-1/2}\pd_2
\end{bmatrix}\]
and
\[E_2^*=-\begin{bmatrix}
									-r_0&										 0&0\\
								 \pd_2&										 0&2^{-1/2}(\pd_3+4(ms+r_0)^{-1}m)\\
-(1+m^2)^{-1}(ms+r_0)m&\pd_3+2(ms+r_0)^{-1}m&2^{-1/2}\pd_2
\end{bmatrix}.\]

Denoting $(1+m^2)^{-1}(ms+r_0)m$ by $\gamma$ and $(ms+r_0)^{-1}m$ by $\lambda$, we arrive at our model expressed as
\begin{equation}\label{Ex4}
\begin{aligned}
&r_0(q_{11}(r_0u_{01}+\gamma u_{03})+q_{12}\pd_3u_{03})+(r_0^2\lambda/\gamma)(\bar{\rho}\pd_t^2u_{01}+r_0\sqrt{\lambda/\gamma}ku_{01})\\&=(r_0\sqrt{\lambda/\gamma})^3(F_{ext,1}-\rho_bF_1),\\
&(\gamma q_{11}-2\lambda q_{12})(r_0u_{01}+\gamma u_{03})+(\gamma q_{12}-2\lambda q_{22})\pd_3u_{03}\\& -\pd_3(q_{12}(r_0u_{01}+\gamma u_{03})+q_{22}\pd_3u_{03})+(1-\gamma\lambda)^2 \bar{\rho}\pd_t^2u_{03}\\&=r_0\sqrt{\lambda/\gamma}(1-\gamma\lambda)^2(F_{ext,3}-\rho_bF_3).
\end{aligned}
\end{equation}

\subsection{Pipe with circular axis and equally spaced layers}\label{lastcase}
We consider a vessel having a circular arc as its central axis (for instance, one can use this model to simulate the circle of Willis that supplies blood to the brain). This results in $\bc'''$ being anti-parallel to $\bc'$ and hence perpendicular to $\be_1$ and $\be_2$. We take a fixed radius $r$ for the vessel. Also,we assume equally spaced layers and hence $|\na G|=1$.
The matrix $M$ is expressed as follows:
\[
	M=\begin{bmatrix}
	1&		 0&0\\
	0&r^{-2}&0\\
	0&		 0&\gamma^{-2}
\end{bmatrix},\]
where $\gamma=1-rc''\cdot\be_1.$

The differential operator matrices $E_2$ and $E_2^*$ in this case are given as
\[E_2=\begin{bmatrix}
									 r&											\pd_2&0\\
-c''\cdot\be_1\gamma&-r^{-1}c''\cdot\be_2\gamma&\pd_3\\
									 0&							2^{-1/2}\pd_3&2^{-1/2}(\pd_2+2rc''\cdot\be_2\gamma^{-1})
\end{bmatrix}\]
and
\[E_2^*=-\begin{bmatrix}
														 -r&			c''\cdot\be_1\gamma&0\\
\pd_2-rc''\cdot\be_2\gamma^{-1}&c''\cdot\be_2r^{-1}\gamma&2^{-1/2}\pd_3\\
															0&										\pd_3&2^{-1/2}(\pd_2-3rc''\cdot\be_2\gamma^{-1})
\end{bmatrix}.\]

With $\lambda_i=c''\cdot\be_i$ for $i=1,2$, we have the model equations for this case as 
\begin{equation}\label{Ex5}
\begin{aligned}
	&(rq_{11}-\gamma\lambda_1q_{12})(ru_{01}+\pd_2u_{02})+(rq_{12}-\gamma\lambda_1q_{22})(\pd_3u_{03}-\gamma(\lambda_1u_{01}+r^{-1}\lambda_2u_{02}))\\&+\bar{\rho}\pd_t^2u_{01}+ku_{01}=F_{ext,1}-\rho_bF_1,\\
	&\lambda_2[(\gamma^{-1}rq_{11}-\gamma r^{-1}q_{12})(ru_{01}+\pd_2u_{02})+(\gamma^{-1}rq_{12}-\gamma r^{-1}q_{22})(\pd_3u_{03}\\
	&-\gamma(\lambda_1u_{01}+r^{-1}\lambda_2u_{02}))]-2^{-1}\pd_3[q_{33}(\pd_3u_{02}+\pd_2u_{03}+2\gamma^{-1}r\lambda_2)]\\
	&-\pd_2[q_{11}(ru_{01}+\pd_2u_{02})+q_{12}(\pd_3u_{03}-\gamma(\lambda_1u_{01}+r^{-1}\lambda_2u_{02}))]\\
	&+r^{-2}\bar{\rho}\pd_t^2u_{02}=r^{-2}(F_{ext,2}-\rho_bF_2),\\
	&-\pd_3[q_{12}(ru_{01}+\pd_2u_{02})+q_{22}(\pd_3u_{03}-\gamma(\lambda_1u_{01}+r^{-1}\lambda_2u_{02}))]\\
	&-2^{-1}(\pd_2-3\gamma^{-1}r\lambda_2)[q_{33}(\pd_3u_{02}+\pd_2u_{03}+2\gamma^{-1}r\lambda_2)]\\
	&+\gamma^{-2}\bar{\rho}\pd_t^2u_{03}=\gamma^{-2}(F_{ext,3}-\rho_bF_3).
\end{aligned}
\end{equation}

\appendix

\section{Description of the tensor $\mathcal{K}$}\label{KTensor}

Let us consider an elastic space weakened by the cylindrical void 
\begin{equation*}
\Om=\{\bx=(\bx',x_3)\in\R^2\times\R:r=|\bx'|=\sqrt{x_1^2+x_2^2}<R\}
\end{equation*}
of radius $R>0$. Assuming the transversal isotropy with the $x_3$-axis of a homogeneous stationary elastic material in $\Xi=\R^3\backslash\bar{\Om}$, we write the equilibrium equations
\begin{equation}
-\na\cdot\s(\bu^m;\bx)=\bnt,\quad\bx\in\Xi,
\label{N1}
\end{equation}
where $\bu^m=(\bu'^m,u_3^m)$ is the three dimensional displacement vector and $\s(\bu)$ the corresponding stress tensor of rank $2$ computed through the Hooke's law. In the Voigt-Mandel notation, the strain and stress 
\[\e=(\e_{11},\e_{22},\sqrt{2}\e_{21},\sqrt{2}\e_{13},\sqrt{2}\e_{23},\e_{33})^T\]
and
\[\s^m=(\s^m_{11},\s^m_{22},\sqrt{2}\s^m_{21},\sqrt{2}\s^m_{13},\sqrt{2}\s^m_{23},\s^m_{33})^T\]
are related by $\s^m=A^m\e$ where
\begin{equation}
A^m=\begin{bmatrix}
\la+2\mu&			\la&	 0&			0&		 0&\al\\
		 \la&\la+2\mu&	 0&			0&		 0&\al\\
			 0&				0&2\mu&			0&		 0&	 0\\
			 0&				0&	 0&2\beta&		 0&	 0\\
			 0&				0&	 0&			0&2\beta&	 0\\
		 \al&			\al&	 0&			0&		 0& \g
\end{bmatrix},
\label{NA}
\end{equation}
$\la\geq0$ and $\mu>0$ are the classical Lam\'e constants in the $\bx'$-plane while the other elastic moduli $\al\geq0$, $\beta>0$ and $\g>0$ are not used here any further.

The particular problem of blood flow requires a description of the interaction of the vessel wall with the surrounding muscle tissue. In other words, we have to find out a relation between the vessel radial dilation
\begin{equation}
\bu^m(\bx)=\bu^w(\bx)=u^w_r \be_r,\bx\in\pd\Om,
\label{N2}
\end{equation}
and the traction
\begin{equation}
\s^m(\bu^m;\bx)\be_r=\s^w(\bu^w;\bx)\be_r,\bx\in\pd\Om,
\label{N3}
\end{equation}
where $\be_r=(r^{-1}\bx',0)$ is the normal vector on $\pd\Om$. Note that the equations (\ref{N1}) do not involve the inertia term $\g_m\pd^2_tu^m(x;t)$ because of the standard reasonable assumption that the Womersley number $W_m$ of the muscle is small in comparison with the Womersley number $W_w$ of the vessel wall. Moreover, vessels are set in so-called vessel beds and in this way are enveloped by a loose cell material in order to prevent gyrations of the wall so that only the radial dilation is passed from the wall to the tissue, cf. (\ref{N2}).

In the general case, the mapping (\ref{N3})$\mapsto$(\ref{N2}) is described with the help of the elasticity Neumann-to-Dirichlet operator which is rather complicated even in our canonical geometry. However, the above-accepted assumption on the low variability of all mechanical fields allows us to employ the asymptotic methods of singularly perturbed elliptic problems \cite{MNP00}.

First of all, the transversal isotropy and the absence of the angular variable $\phi\in[0,2\pi)$ in (\ref{N2}) prove that $u^m_\phi=\bu^m\cdot\be_\phi=0$ in $\Xi_R$. Furthermore, the low variability along the variable $z$ eliminates derivatives in $z$ in the Cauchy formulas (\ref{Cauchy}) as well as in the equations (\ref{N1}) which take the form
\[-\frac{\pd}{\pd x_1}\s_{j1}(\bu^m)-\frac{\pd}{\pd x_2}\s_{j2}(\bu^m)=0\mbox{ in }\Xi_R,j=1,2,3\]
or, in view of (\ref{NA}), become the plane elasticity system
\begin{equation}
-\mu\D_{\bx'}u^m_j-(\la+\mu)\frac{\pd}{\pd x_j}\left(\frac{\pd u^m_1}{\pd x_1}+\frac{\pd u^m_2}{\pd x_2}\right)=0, j=1,2,
\label{N4}
\end{equation}
\begin{equation}
-\beta\D_{\bx'}u^m_3=0,\bx'\in\R^2\backslash\Om.
\label{N5}
\end{equation}
At the same time, (\ref{N2}) reads componentwise as follows:
\begin{equation}
u^m_r(\bx')=u^w_r(s), u^m_\phi(\bx')=0,\bx'\in\Gin,
\label{N6}
\end{equation}
\begin{equation}
u^m_3(\bx')=0,\bx'\in\Gin.
\label{N7}
\end{equation}
From (\ref{N5}) and (\ref{N7}) we derive that
\begin{equation}
u^m_3(\bx')=0,\bx'\in\R^2\backslash\Om.
\label{N8}
\end{equation}
According to (\ref{N6}), the displacement field ${\bu^m}'$ is axisymmetric and, therefore, in the polar coordinates $(r,\phi)$, we have
\begin{align}\label{NS}
	u^m_r(r,\phi)=\frac{a}{r}, u^m_\phi(r,\phi)=0,\\
	\s^m_{rr}({u^m}';r,\phi)=-2\mu\frac{a}{r^2},\s^m_{\phi\phi}({u^m}';r,\phi)=2\mu\frac{a}{r^2},\\
	\s^m_{r\phi}({u^m}';r,\phi)=0.
\end{align}
Finally, (\ref{N6}) gives $a=Ru^w_r(s)$ so that
\[\s^w(u^w;s,\phi)e_r=-\frac{2\mu}{R}u^w_r(s)e_r.\]
This relation gives the tensor $\mathcal{K}$ in (\ref{obc}) while 
\[k=\frac{2\mu}{R}h^{-1}\mbox{ in (\ref{k00})}\] because $\mathcal{K}$ has the factor $h$ in (\ref{obc}).

In this way, we need to assume that the elastic characteristics of the wall and the muscle are in the relation $h^{-2}:1$.

\section{Leading order terms in the expansions of the basis vectors, Christoffel symbols and the metric tensor}\label{AppExp}

In this section, we present the values of several geometric quantities used in our case. We use the big O notation to express the corresponding values in the new radial variable $\xi$. Here, we let $R=R(\bx)=|\na G(\bx)|^{-2}\na G(\bx)$ and $\alpha=\alpha(\theta,s)=((1-r(s)\bc''(s)\cdot\be_1(\theta,s))^2+r'(s)^2)^{-1}$ in order to have relatively compact expressions. For the same reason, we write all the functions without their respective arguments. 

We first describe the contravariant basis vectors.
\begin{align*}
	\bx_1&=R,\\
	\bx_2&=r\be_2+\int^{n}_{0}(\pd_\theta R)\dif r=r\be_2+O(h),\\
	\bx_3&=r'\be_1+(1-r\bc''\cdot\be_1)\bc'+\int^{n}_{0}(\pd_sR)\dif r=r'\be_1+(1-r\bc''\cdot\be_1)\bc'+O(h).
\end{align*}
The metric tensor components are calculated using the relation $g_{ij}=\bx_i\cdot\bx_j$, thereby providing us with
\begin{align*}
	g_{11}&=|\na G|^{-2},\\
	g_{12}&=g_{21}=g_{13}=g_{31}=0,\\
	g_{22}&=|r\be_2+\int^{n}_{0}(\pd_\theta R)\dif r|^2=r^2+O(h),\\
	g_{33}&=|r'\be_1+(1-r\bc''\cdot\be_1)\bc'+\int^{n}_{0}(\pd_sR)\dif r|^2=r'^2+(1-r\bc''\cdot\be_1)^2+O(h),\\
	g_{23}&=[r\be_2+\int^{n}_{0}(\pd_\theta R)\dif r]\cdot[r'\be_1+(1-r\bc''\cdot\be_1)\bc'+\int^{n}_{0}(\pd_sR)\dif r]=0+O(h).
\end{align*}

Using the above expressions, the reciprocal basis vectors are calculated according to the formula $\bx^i=g^{ij}\bx_j$ where the inverse metric tensor satisfies $g^{ij}g_{jk}=\delta^i_k$. This results in
\begin{align*}
	\bx^1&=g^{1i}\bx_i=\na G,\\
	\bx^2&=g^{2i}\bx_i=r^{-1}\be_2+O(h),\\
	\bx^3&=g^{3i}\bx_i=\alpha[r'\be_1+(1-r\bc''\cdot\be_1)\bc']+O(h).
\end{align*}

The Christoffel symbols are calculated in accordance with the identities $\G^i_{jk}=\bx^i\cdot\bx_{jk}$, to be as below. Note that we can use the symmetry $\G^i_{jk}=\G^i_{kj}$ to find the full set of $27$ Christoffel symbols in our case.
\begin{align*}
	\G^1_{11}&=\na G\cdot\na RR=-R\cdot(\na\na G)R,\\
	\G^1_{12}&=\na G\cdot \pd_\theta R,\\
	\G^1_{13}&=\na G\cdot \pd_sR,\\
	\G^1_{22}&=-r\na G\cdot\be_1+O(h),\\
	\G^1_{33}&=\na G\cdot[r''\be_1-(2r'\bc''\cdot\be_1+r\bc'''\cdot\be_1)\bc'+(1-r\bc''\cdot\be_1)\bc'']+O(h),\\
	\G^1_{23}&=\na G\cdot[r'\be_2-r\bc''\cdot\be_2\bc']+O(h),\\
	\G^2_{11}&=r^{-1}\be_2\cdot\na RR+O(h),\\
	\G^2_{12}&=r^{-1}\be_2\cdot \pd_\theta R+O(h),\\
	\G^2_{13}&=r^{-1}\be_2\cdot \pd_sR+O(h),\\
	\G^2_{22}&=0+O(h),\\
	\G^2_{33}&=r^{-1}(1-r\bc''\cdot\be_1)\bc''\cdot\be_2+O(h),\\
	\G^2_{23}&=r^{-1}r'+O(h),\\
	\G^3_{11}&=\alpha[r'\be_1+(1-r\bc''\cdot\be_1)\bc']\cdot\na RR+O(h),\\
	\G^3_{12}&=\alpha[r'\be_1+(1-r\bc''\cdot\be_1)\bc']\cdot \pd_\theta R+O(h),\\
	\G^3_{13}&=\alpha[r'\be_1+(1-r\bc''\cdot\be_1)\bc']\cdot \pd_sR+O(h),\\
	\G^3_{22}&=-\alpha rr'+O(h),\\
	\G^3_{33}&=\alpha[\{r''+(1-r\bc''\cdot\be_1)\bc''\cdot\be_1\}r'-(2r'\bc''\cdot\be_1+r\bc'''\cdot\be_1)(1-r\bc''\cdot\be_1)]+O(h),\\
	\G^3_{23}&=-\alpha r\bc''\cdot\be_2(1-r\bc''\cdot\be_1)+O(h).
\end{align*}

\bibliographystyle{plain}
\bibliography{bibl}

\end{document}